\documentclass[12pt,reqno]{article}

\usepackage[usenames]{color}
\usepackage{amssymb}
\usepackage{amsmath}
\usepackage{amsthm}
\usepackage{amsfonts}
\usepackage{amscd}
\usepackage{graphicx}
\usepackage{diagbox}

\usepackage[colorlinks=true,
linkcolor=webgreen,
filecolor=webbrown,
citecolor=webgreen]{hyperref}

\definecolor{webgreen}{rgb}{0,.5,0}
\definecolor{webbrown}{rgb}{.6,0,0}

\usepackage{color}
\usepackage{fullpage}
\usepackage{float}

\usepackage{graphics}
\usepackage{latexsym}
\usepackage{epsf}
\usepackage{breakurl}

\DeclareMathOperator{\per}{per}

\begin{document}

\theoremstyle{plain}
\newtheorem{theorem}{Theorem}
\newtheorem{corollary}[theorem]{Corollary}
\newtheorem{lemma}[theorem]{Lemma}
\newtheorem{proposition}[theorem]{Proposition}

\theoremstyle{definition}
\newtheorem{definition}[theorem]{Definition}
\newtheorem{example}[theorem]{Example}
\newtheorem{conjecture}[theorem]{Conjecture}

\theoremstyle{remark}
\newtheorem{remark}[theorem]{Remark}

\title{Avoidance of split overlaps}

\author{Daniel Gabric and Jeffrey Shallit \\
School of Computer Science \\
University of Waterloo \\
Waterloo, ON  N2L 3G1 \\
Canada \\
{\tt dgabric@uwaterloo.ca} \\
{\tt shallit@uwaterloo.ca} 
}

\date{}

\maketitle

\begin{abstract}
We generalize Axel Thue's familiar definition of overlaps in words,
and show that there are no infinite words containing split
occurrences of these generalized overlaps.  Along the way we
prove a useful theorem about repeated disjoint occurrences in words --- an interesting natural variation on the classical de Bruijn sequences.
\end{abstract}

\section{Introduction}

In this paper,
we are concerned with words over a finite alphabet $\Sigma$
of cardinality $k \geq 1$; more specifically, avoiding certain
kinds of repetitions in them.

Two kinds of repetitions that have been studied for more than
a hundred years are squares and overlaps \cite{Thue:1912,Berstel:1995}.
A {\it square} is a finite nonempty word of the form $xx$ 
(such as the English word {\tt murmur}).
Another type of repetition is the $\alpha$-power.  We say a word $w$
is an $\alpha$-power, for $\alpha = p/q$, a rational number,
if $|w| = p$ and $w$ has period $q$.   (We say a word $w$ has period
$q \geq 1$ if $w[i] = w[i+q]$ for all $i$ for which this makes sense.)
Thus {\tt alfalfa} is a $(7/3)$-power.
A word $y$ is a {\it factor\/} of a word $w$ if
$w = xyz$ for words $x, z$ (possibly empty).  When we speak about a word
``avoiding $\alpha$-powers'', we mean it has no factor that is a
$\beta$-power, for all $\beta \geq \alpha$.  
The smallest period of a word $w$ is sometimes called {\it the\/} period,
and is written $\per(w)$.

An {\it overlap} is a finite word of the form
$axaxa$ for $a$ a single letter, and $x$ a (possibly) empty word, such
as the French word {\tt entente}.
An overlap can be viewed as just slightly more than a square: it
consists of two repetitions of a nonempty word $w$, followed by the first letter
of $w$.

The term ``overlap'' comes from the following ``folk'' observation:  say
two distinct occurrences of a length-$n$ factor $x$ in $w$, say
$x = w[i..i+n-1] = w[j..j+n-1]$ with $i < j$, ``overlap each other''
if $0 < j-i < n$.  
\begin{proposition}
If $w$ contains two distinct occurrences of $x$ that overlap each other,
then $w$ contains an overlap.
\label{prop1}
\end{proposition}
\begin{proof}
Define $y = w[i..j-1]$, $t = w[j..i+n-1]$, and $z = w[i+n..j+n-1]$ and examine Figure~\ref{fig1}.
\begin{figure}
    \centering
    \includegraphics{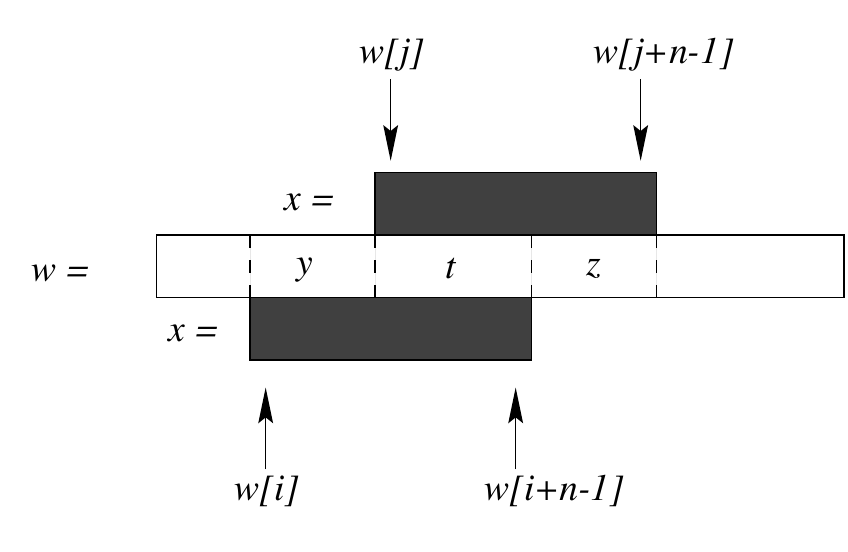}
    \caption{Overlapping factors}
    \label{fig1}
\end{figure}
Each of these words is nonempty, and $x = yt = tz$.  By the Lyndon-Sch\"utzenberger
theorem \cite{Lyndon&Schutzenberger:1962}, it follows that
there exist words $u, v$ with $u$ nonempty, and an integer $e\geq 0$, such that
$y = uv$, $t = (uv)^e u$, and $z = vu$.  
Thus $w[i..j+n-1] = ytz = (uv)^{e+2} u$, which contains an overlap.
\end{proof}

This suggests the following natural generalization of overlap:
a {\it $t$-overlap\/} is a word of the form $x x x'$, where
$x$ is a nonempty word of length at least $t$, and 
$x'$ is the first $t$ letters of $x$.
For example, the unfamiliar English word
{\tt prelinpinpin} contains a suffix that is a $2$-overlap.
Note that a $0$-overlap is a square, and a $1$-overlap is an ordinary overlap.
Of course, a $t$-overlap contains a $t'$-overlap for all $t' < t$.

Thue proved \cite{Thue:1912,Berstel:1995} that one can avoid
overlaps over any alphabet containing at least two letters.
Here by ``avoid'' we mean ``there exists an infinite word
containing no overlaps'' or, equivalently, ``there exist infinitely
many finite words containing no overlaps''.  Since
every $t$-overlap contains a $1$-overlap, Thue's construction also
shows it is possible
to avoid $t$-overlaps over any alphabet with at least two letters.

So instead, in this paper
we consider split occurrences of repetitions.
A {\it split occurrence} of a repetition is a word of the form $xyz$,
where $xz$ forms the repetition.  For example, the English word
{\tt contentment} contains a split occurrence of the $2$-overlap
{\tt ntentent}, which arises from the factor $xyz$,
where $x = {\tt ntent}$, $y = {\tt m}$, $z = {\tt ent}$.

It follows from known results that there exist
infinite words avoiding split occurrences
of $\alpha$-powers, for any rational number $\alpha > 2$
\cite{Mousavi&Shallit:2013}.  To see this,
take the alphabet size $k$ sufficiently large that there exists
an infinite word $\bf w$ over $\Sigma_k = \{0, 1,\ldots, k-1 \}$ avoiding $\alpha/2$ powers.
(By Dejean's theorem
\cite{Dejean:1972,Currie&Rampersad:2011,Rao:2011} this is possible.)
Suppose $xyz$ is a factor of $\bf w$
that is a split occurrence of a $\beta$ power for
$\beta \geq \alpha > 2$. Then clearly either $x$ is a
$\geq \beta/2$ power or $z$ is, a contradiction.

In contrast, in this paper we show that no matter what the alphabet
size is, there are no infinite words
avoiding split occurrences of $t$-overlaps.
Our main theorem is the following.
\begin{theorem}
Let $w$ be a word over a $k$-letter alphabet avoiding split
occurrences of $t$-overlaps.  Then $w$ is finite.
\label{main}
\end{theorem}

We also investigation the avoidance of {\it reversed split} repetitions.
A reversed split repetition is a word of the form $xyz$, where
$zx$ forms the repetition.  For example, the English word
{\tt independent} contains a reversed split overlap:  it has the
factor $xyz$, where $x = {\tt nde}$, $y = {\tt p}$, and
$z = {\tt ende}$, giving the overlap $zx = {\tt endende}$.  

\section{Some useful results on primitive words and bordered words}

We call a nonempty word $w$
{\it primitive\/} if $w$ cannot be written in the form
$x^k$ for an integer $k \geq 2$; see, for example, \cite{Domosi&Ito:2015}.

\begin{lemma}
Let $A_k(n,p)$ denote the number of length-$n$
words over $\Sigma_k$ with smallest period $p$, and let
$\psi_k(n)$ denote the number of primitive words over
$\Sigma_k$.  Then $A_k(n,p) = \psi_k(p)$ for $1 \leq p \leq {n\over 2} + 1$.
\label{lem2}
\end{lemma}

\begin{proof}
We claim that every length-$n$ word $w$ with shortest period $p$ can be written
in the form $w = x^i x'$, where $x'$ is a prefix of $x$ and
$|x| = p$ and $x$ primitive.  For if $x$ were not primitive,
say $x = y^j$ for some $j \geq 2$, then $p$ could not be the
shortest period.

We now claim that if $x$ is primitive and $1 \leq p \leq {n\over 2} + 1$,
then $w = x^{n/p}$ has shortest period $p$.  Suppose to the contrary
that $w$ has shortest period $q < p$.   Since
$n \geq p + q - 1$, by the Fine-Wilf theorem \cite{Fine&Wilf:1965},
$w$ also has the period
$\gcd(p,q)$.  If $q$ divides $p$, then $x$ was not primitive, a
contradiction.  Otherwise $\gcd(p,q) < q$, a contradiction.
\end{proof}

A {\it border} of a word $w$ is a nonempty word $x$, $x \not = w$,
such that $x$ is both a prefix and suffix of $w$.  Thus
{\tt entanglement} has the border {\tt ent}.  If a word has
a border, it is called {\it bordered}, and otherwise it is called
{\it unbordered}.  It is easy to see
that if a word of length $n$ has a border, it must have a border
of length $\leq n/2$.

\begin{lemma}
For $k \geq 2, n \geq 1$,
there are at least $k^n(1 - 1/k - 1/k^2)$  unbordered words of 
length $n$ over a $k$-letter alphabet.
\label{lem4}
\end{lemma}

\begin{proof}
Let $u_k (n)$ denote the number of unbordered words of length
$n$ over a $k$-letter alphabet.  It follows from the recurrence for
$u_k(n)$ given in \cite{Nielsen:1973} that $u(k,n)$ is a polynomial of degree $n$ in $k$.   By explicit computation of these polynomials for $n = 1, 2, \ldots, 12$,
we can easily verify the inequality $u_k(n) \geq k^n(1-1/k-1/k^2) \geq k^n(1-1/k-1/k^2)$ for $n \leq 12$.    In particular,
$u_{k} (12) = k^{12} - k^{11} - k^{10} + k^6 + k^5 - k^2$.

Now assume $n > 12$.   For each unbordered word $w$ of length $12$, write $w = xz$ with $|x| = |z| = 6$, and consider the words $xyz$ of length $n$, where $y$ is an arbitrary word of length $n-12$.  There are $u_k(12) k^{n-12}$ such words.   Each such word is unbordered, unless it has a border of length $i$ for $6 < i \leq n/2$.
But the total number of words with 
border length $i$ satisfying $6 < i \leq n/2$ is at most $$
k^{n-7} + k^{n-8} + \cdots + k^{n/2} 
\leq (k^{n-6} - 1)/(k-1) .$$

Therefore, there are least
$$u_k (12) k^{n-12} - (k^{n-6} -1)/(k-1) = 
 k^n(1-1/k - 1/k^2 + 1/k^6 + 1/k^7 - 1/k^{10}) - (k^{n-6} -1)/(k-1)$$
unbordered words of length $n$ for 
$n > 12$.  Since $k^n/k^6 \geq (k^{n-6} -1)/(k-1)$ and $k^{n-7} \geq k^{n-10}$,
the desired bound follows.
\end{proof}

\section{Disjoint occurrences}

Let $\Sigma_k = \{ 0, 1, \ldots,k-1 \}$ be an alphabet of $k \geq 1$
letters.
It is known that for every $k \geq 1$ and $n \geq 1$,
there exists a word of length $k^n + n-1$ that contains every
length-$n$ word exactly once as a factor; such words are
called {\it de Bruijn words} of order $n$; see
\cite{deBruijn:1946,deBruijn:1975}.   This bound of $k^n + k-1$
is optimal, because from the pigeonhole
principle, it follows that if $w$ is a word of length
$\geq k^n + n$, then $w$ must contain at least two different
occurrences of some word $x$ of length $n$.  

However, these two different occurrences of $x$ could overlap each
other in $w$.   If two distinct occurrences do not overlap,
we say they are {\it disjoint}.)

If we insist on having two disjoint
occurrences, we get a different bound.
For example, there are binary words of length $7$
that do not contain two disjoint occurrences of
the same length-$2$ word, such as $0111000$.  
Let us define $C(k,n)$ to be the length of the longest
word over $\Sigma_k$ having the property that there are no
two disjoint occurrences of the same word.
By considering disjoint occurrences of length-$n$ blocks,
the pigeonhole principle easily gives the bound
$C(k,n) < n(k^n + 1)$.  
We now obtain some better bounds on $C(k,n)$.

We need a lemma.

\begin{lemma}
Let $x,w$ be words with $|x| = n$.  
Suppose $w$ contains $m$ occurrences of $x$, but
not two or more disjoint occurrences.   Then
$m \leq \lceil n/\per(x) \rceil$. Furthermore, for
each individual $x$, this upper bound is achievable.
\label{lem1}
\end{lemma}

\begin{proof}
Let $w$ contain the maximum possible number of overlapping
occurrences of the length-$n$
word $x$, and no disjoint occurrences of $x$.
Let $d$ be the shortest distance between two consecutive occurrences
of $x$ in $w$.   If there are $m$ overlapping occurrences,
then the last occurs at distance at least $d(m-1)$ from the first.
If $d(m-1) \geq n$, then the last occurrence does not overlap the first,
so $d(m-1) < n$.  It follows that $m < n/d + 1$, and since
$t$ is an integer, we have $m \leq \lceil n/d \rceil$.

We now show that $d = \per(x)$.  Two overlapping occurrences of
$x$ with the shortest distance between them correspond to writing
$x = yt = tz$ for some $y, t, z$ (with $t$ the overlap), with
$1 < |t| < n$, and minimizing $|y|$; see the diagram.
Now, from the
Lyndon-Sch\"utzenberger theorem \cite{Lyndon&Schutzenberger:1962},
it follows that there exist $u,v$ with $u$ nonempty and
an integer $e \geq 0$ such that
that $y = uv$, $t = (uv)^e u$, and $z = vu$.
Hence $y = uv$ is a period of $x$; to minimize $y$ we take
$y$ to be the shortest period of $x$.

We have now shown that $m \leq \lceil n/\per(x) \rceil$.
It remains to see that this bound is always achievable.
Let $y$ be the shortest period of $x$, and write
$x = y^f u$, where $u$ is a nonempty prefix of $y$, possibly equal
to $y$ itself.   Then $y = uv$ for some (possibly empty) $v$.
Consider the word $w = (uv)^{2f} u$; 
it is easy to see that
$x = (uv)^f u$ overlaps itself at least $f+1$ times in this $w$.
Since $f|y| < n \leq (f+1) |y|$, it follows that
$f + 1 = \lceil n/\per(x) \rceil$.
\end{proof}

\begin{theorem}
We have 
$$C(k,n) \leq \left( \sum_{w \in \Sigma_k^n} \left\lceil {n \over {\per(w)}} 
\right\rceil \right) + n-1.$$
\label{thm3}
\end{theorem}

\begin{proof}
Let $w$ be a longest word having no disjoint occurrences of
the same length-$n$ factor.  Let us now count the number of occurrences
of each length-$n$ factor $x$ in $w$.
By Lemma~\ref{lem1}, $w$ can contain at most
$\lceil n/\per(x) \rceil$ occurrences of $x$.
Thus, in the worst case, $w$ can have at most
$\sum_{x \in \Sigma_k^n} \lceil {n \over {\per(x)}} \rceil$
total occurrences of length-$n$ words.  
Thus the word can be of length at most
$\left( \sum_{x \in \Sigma_k^n} 
\lceil {n \over {\per(x)}} \rceil \right) + n-1$.
\end{proof}

\begin{corollary}
For $k \geq 2$
we have $C(k,n) \leq k^n( 1 + 1/k + 1/k^2  ) +n (k^{n/2 + 1}-1)/(k-1 ) + n-1 $.
\end{corollary}

\begin{proof}
We split the sum
$\sum_{x \in \Sigma_k^n} \lceil {n \over {\per(x)}} \rceil$
into three parts:  one where
$\per(x) \leq n/2$, one where $n/2 < \per(x) < n$,
and one where $\per(x) = n$.

From Lemma~\ref{lem2} above, the number of length-$n$
words $x$ with smallest period $p \leq n/2$ is
$\psi(k,p)$, the number of primitive words of length
$p$ over a $k$-letter alphabet.   
Write $A = \sum_{1 \leq p \leq n/2} \psi(k,p)$
and $B = \sum_{1 \leq p \leq n/2} \psi(k,p) \lceil n/p \rceil$.
It is known that
$\psi(k,n) = \sum_{d|n} \mu(d) k^{n/d}$, where
$\mu$ is the M\"obius function from number theory (see, e.g., \cite[p.~245]{Domosi&Ito:2015}),
but the much weaker bound $\psi(k,n) \leq k^n$ suffices for our
purposes here.   
Thus $B \leq n (k + k^2 + \cdots + k^{n/2}) \leq n (k^{n/2+1} -1)/(k-1)$.  

The number of words with period $n$ is $u_k(n)$, the number
of unbordered words of length $n$.   From Lemma~\ref{lem4} we have
$u_k(n)  \geq  k^n(1 - 1/k - 1/k^2 )$.  
Thus we have
\begin{align}
\sum_{x \in \Sigma_k^n} \left\lceil {n \over {\per(x)}} \right\rceil
&=  B + 2(k^n - A - u_k (n)) + u_k(n)  \label{formula} \\
&\leq  2 k^n  - u_k (n) + B \nonumber \\
&\leq  2k^n - k^n(1 - 1/k - 1/k^2 ) + B  \nonumber \\
& \leq k^n( 1 + 1/k + 1/k^2 ) + n (k^{n/2+1} -1)/(k-1), \nonumber
\end{align}
from which the result follows.
\end{proof}

\begin{theorem}
\leavevmode
\begin{itemize}
\item[(a)] $C(1,n) = 2n-1$ for $n \geq 1$;
\item[(b)] $C(k,1) = k$ for $k \geq 1$;
\item[(c)] $C(k,2) = k^2 + k + 1$ for $k \geq 1$;
\item[(d)] $C(k,3) = k^3 + k^2 + k + 2$ for $k \geq 1$.
\end{itemize}
\end{theorem}

\begin{proof}
\leavevmode
\begin{itemize}
\item[(a)]  A unary word of length $2n$ has two disjoint
length-$n$ occurrences.
\item[(b)]  A word of length $k+1$, by the pigeonhole principle,
has two occurrences of a single letter.
\item[(c)]  Take a de Bruijn word of order $2$ over a 
$k$-letter alphabet; it has length $k^2 + 1$.
Replace each occurrence of $aa$ with $aaa$;
such a replacement clearly does not introduce
any disjoint occurrences.  The resulting word has
length $k^2 + k + 1$.  This gives the lower bound.
For the upper bound, we use Theorem~\ref{thm3}.
All length-$2$ words have period $2$, except those
of the form $aa$, which have period $1$.  
Then the sum in Theorem~\ref{thm3} gives the upper bound.
\item[(d)]  For the upper bound, we note that
all length-$3$ words have period $3$, except
that $aaa$ has period $1$ and $aba$, with $a \not=b$,
has period $2$.  The sum in Theorem~\ref{thm3} then gives
$k^3 + k^2 + k + 2$.  

To see the lower bound, we need some terminology and a lemma.

A function $f: \Sigma_k^n \to \Sigma_k$ is said to be a \emph{feedback
function}. A feedback function $f$ is said to be \emph{non-singular}
if the function $F: \Sigma_k^n \to \Sigma_k^n$ defined as $F(a_1a_2\cdots
a_n) = a_2\cdots a_n f(a_1a_2\cdots a_n)$ is one-to-one.

A \emph{universal cycle} for a set of words $S\subseteq \Sigma_k^n$ is a
length-$|S|$ word that, when considered circularly, contains every word in
$S$ as a factor. A non-singular feedback function partitions $\Sigma_k^n$
into sets $S_1,S_2,\ldots, S_m$, each having a corresponding universal
cycle. For each word $w=w_1w_2\cdots w_n \in S_i$, for some $1\leq i \leq m$, we have that $w_2w_3\cdots w_nf(w) \in S_i$ and $w$ has a corresponding
word $v = v_1v_2\cdots v_n\in S_i$ such that $w = v_2v_3\cdots v_n f(v)$.
The \emph{cycle representative} of a set $S$ is the lexicographically
least word in $S$. It is well known that de Bruijn words can be
constructed by joining the universal cycles in a specific way, sometimes with use of a successor rule. A \emph{successor rule} is a feedback function that determines
the next symbol in a de Bruijn word using the previous $n$ symbols.
\begin{lemma}
For all $k\geq 2$, there exists a $k$-ary de Bruijn word of order $3$ that contains either $abab$ or $baba$ for all $a\neq b$ where $a,b\in \Sigma_k$.
\end{lemma}
\begin{proof}
Consider the feedback function $f:\Sigma_k^3 \to \Sigma_k$ defined by
$f(a_1a_2a_3) = a_1 + a_2 - a_3$. We will show that the function $F(a_1a_2a_3) = a_2a_3f(a_1a_2a_3)$ is one-to-one. Suppose there exist two words $a_1a_2a_3$ and $b_1b_2b_3$ such that $F(a_1a_2a_3) = F(b_1b_2b_3)$. Then we would have that $a_2a_3(a_1+a_2-a_3) = b_2b_3(b_1+b_2 - b_3)$. But this implies that $a_2 = b_2$,$a_3 = b_3$, and $a_1+a_2-a_3 = b_1+b_2-b_3$. These three equations imply $a_1 = b_1$. Now we have $a_1a_2a_3 = b_1b_2b_3$. Therefore $F$ is one-to-one.

We now define a feedback function $g$ based $f$ and we show $g$ is a successor rule. Let
$a_1a_2a_3 \in \Sigma_k^3$. Let $\tau(a_2a_3)$ be an increasing sequence
of symbols $c\in \Sigma_k$ such that $a_2a_3c$ is a cycle representative
of some set in the partition of $\Sigma_k^3$ by $f$. If $0$ is in
$\tau(a_2a_3)$ and $a_2a_3c \neq 000$, then prepend $f(0a_2a_3)$
to the sequence. If $0$ is not in $\tau(a_2a_3)$ and $\tau(a_2a_3)$
is nonempty, then prepend $0$ to the sequence. Let $t_0,t_1,\ldots,
t_{p-1}$ be the sequence $\tau(a_2a_3)$. Let $g:\Sigma_k^3 \to \Sigma_k$
be a feedback function defined as follows:
$$ g(a_1a_2a_3) =  
\begin{cases}
      t_{(j+1)\text{ mod }p}, &\text{ if }f(a_1a_2a_3) = t_j\text{ for some }j \in
      \{0,1,\ldots , p-1\}; \\
f(a_1a_2a_3), &\text{ otherwise.}
\end{cases}
$$
By \cite[Theorem 4.3]{GSWW2020} we
have that $g$ is a successor rule. We now argue that $g(aba)
= b$ for all $a,b\in \Sigma_k$ with $a<b$. Since $f(aba)
= a + b - a = b$, it suffices to show that $\tau(ba)$
is empty. Suppose that $\tau(ba)$ is nonempty. Then
there exists a $d\in \Sigma_k$ such that $bad$ is a cycle
representative of some set $S'$ in the partition of $\Sigma_k^3$
by $f$. Consider the word $adf(bad) \in S'$. Since $a < b$, we
have that $adf(bad)$ is lexicographically smaller than $bad$. Thus $bad$ cannot be a cycle
representative. So  $\tau(a_2a_3)$ is empty.
\end{proof}

We can now continue with the proof of the
lower bound.  Suppose $abab$ is the occurrence; we can then insert
$ab$ immediately after its occurrence.  We can also
insert $aa$ after the unique occurrence of $aaa$
for each letter $a$.   This introduces no disjoint
occurrences, but adds $k(k-1) + 2k$ letters
to the de Bruijn word of length $k^3 + 2$,  thus matching
the upper bound.
\end{itemize}
\end{proof}

Computing the exact value of $C(k,n)$, even for $k$ and $n$ seems like a difficult
problem.   In Table~\ref{tab1} below we give the first few
values of this function, obtained by brute force of the solution space.
\begin{table}[H]
\begin{center}
\begin{tabular}{|c|ccccccc|}
\hline
\diagbox{$k$}{$n$} & 1 & 2 & 3 & 4 & 5 & 6 & 7 \\
\hline
1 & 1 & 3 & 5 & 7 & 9 & 11 & 13 \\
2 & 2 & 7 & 16 & 32 & 59 & 110 & $\geq 192$  \\
3 & 3 & 13 & 41 & & & &\\
4 & 4 & 21 & 86 & & & &\\
5 & 5 & 31 & & & & & \\
\hline
\end{tabular}
\end{center}
\caption{Values of $C(k,n)$}
\label{tab1}
\end{table}
Words achieving the bounds in Table~\ref{tab1} are given below:
\begin{table}[H]
\begin{center}
\begin{tabular}{c|c|l}
$k$ & $n$  & Word achieving $C(k,n)$ \\
\hline
2 & 2 & $0001110$ \\
2 & 3 & $0000010101111100$ \\
2 & 4 & $01010100100110110111111100000001$ \\
2 & 5 & $00000000010001000110011001110100101010101101101111111110000$ \\
2 & 6 & {\tiny 00000000000100001000011000110001110011100111101000101001010010110010011011011010101010111011101111111111100000} \\
3 & 2 & $0001021112220$ \\
3 & 3 & $00000101011002020210220121212222211111200$ \\
4 & 2 & $000102031112132223330$ \\
4 & 3 & {\footnotesize 00000101011002020210220030303103201203301302311111212122113131321331232323333322222300} \\
5 & 2 & $0001020304111213142223243334440$ 
\end{tabular}
\end{center}
\caption{Words achieving the bounds in Table~\ref{tab1}}
\end{table}
For all of the entries in this table, except $(4, 3)$,
the word given is guaranteed
to be the lexicographically least.

The value $C(2,6) = 110$ and the associated
lexicographically least string, and the bound $C(2,7) \geq 192$ were
computed by Bert Dobbelaere, who kindly allowed us to quote them here.

\section{Split occurrences of $t$-overlaps}

We now can prove Theorem~\ref{main}, the main result of this paper.   We do so by finding an explicit bound on the length of the longest word avoiding split overlaps.
Let $n \geq 0$ and $k \geq 1$ be fixed integers.
Define $S(k,t)$ (resp., $R(k,t)$) to be the length of the longest word
over a $k$-letter alphabet containing no occurrences
of split $t$-overlaps (resp., reversed split $t$-overlaps).

\begin{theorem}
We have
\begin{itemize}
\item[(a)] $S(k,t) \leq C(k, C(k,t)+ 1) $;
\item[(b)] $S(k,0) = k$;
\item[(c)] $S(k,1) \leq k^{k+1} + k-1$;
\item[(d)] $S(1,t) = 3t-1$ for $t \geq 1$; 
\end{itemize}
and the same bounds hold for $R(k,t)$.
\end{theorem}

\begin{proof}
We prove the results only for split overlaps; exactly the same arguments can be used for reversed split overlaps.

\medskip

\noindent(a) Let $|w| \geq C(k, C(k,t) + 1) + 1$.  Then
$w$ contains at least two disjoint occurrences
of some factor $x$ of length $C(k,t) + 1$.   Write $w = pxqxr$.  Then 
$x$ itself contains two disjoint
occurrences of some factor $y$ of length $t$.  Write
$x = syuyv$.   Then
$w = p syuyv q syuyv r$.   Now $w$ contains the factor
$yuyvqsyuy$ and so the split $t$-overlap $yuy\cdot uy$.  It
therefore follows that $S(k,t) \leq C(k,C(k,t)+1)$, as desired.

\medskip

\noindent(b) For $t = 0$, we can take $C(k,t) = k$.   For if a word $w$ is of length
at least $k+1$, it must contain two repeated letters, say
$w = xayaz$, and hence the split square $a\cdots a$.

\medskip

\noindent(c) For $t = 1$, we have $C(k,t) \leq k^{k+1} + k-1$.    We can
use the argument in (a), but with a small twist.
Consider the factors of length $k+1$ in a word $w$ of length at
least $k^{k+1} + k$.  There are
at least $k^{k+1} + 1$ of these factors, and by the pigeonhole principle,
some factor $x$ of length $k+1$ appears at least twice in
$w$.
If these two occurrences of $x$ overlap in $w$, we are
already done, because they contain an overlap right there
by Proposition~\ref{prop1}.
Otherwise, write $w = s x t x u$ for some $s, t, u$.
Now $x$ is of length $k+1$, so again by the pigeonhole principle,
some letter $a$ is repeated in
$x$.  Write $x = paqar$ for some words $p, q, r$.
Putting this all together, we have $w = s paqar t paqar u$.
Consider the factor $aqartpaqa$.    It has the split
$1$-overlap $aq\cdots aqa$.

\medskip

\noindent(d) Easy.  Left to the reader.
\end{proof}

Table~\ref{tab2} gives the values of $S(k,t)$ we have computed
by brute force.
\begin{table}[H]
\begin{center}
\begin{tabular}{|c|ccccc|}
\hline
\diagbox{$k$}{$t$} & 0 & 1 & 2 & 3 & 4  \\
\hline
1 & 1 & 2 & 5 & 8 & 11 \\
2 & 2 & 4 & 12 & 47 & \\
3 & 3 & 9 & $\geq97$ & &   \\
4 & 4 & 31 & & & \\
5 & 5 & $\geq 100$ &&& \\
\hline
\end{tabular}
\end{center}
\caption{Values of $S(k,t)$}
\label{tab2}
\end{table}

Words achieving the nontrivial bounds in Table~\ref{tab2} are
given below:
\begin{table}[H]
\begin{center}
\begin{tabular}{c|c|l}
$k$ & $t$  & Lexicographically least word achieving $S(k,t)$ \\
\hline
2 & 1 & $0011$ \\
2 & 2 & $000110100111$ \\
2 & 3 & $00111010100001010011101000011111000011010001110$\\ 
3 & 1 & $012021012$ \\
4 & 1 & $0120321301231013210203123021031$
\end{tabular}
\end{center}
\caption{Lexicographically least word achieving the bounds in Table~\ref{tab2}}
\end{table}

Table~\ref{tab3} gives the values of $R(k,t)$ we have 
computed by brute force.
\begin{table}[H]
\begin{center}
\begin{tabular}{|c|ccccc|}
\hline
\diagbox{$k$}{$t$} & 0 & 1 & 2 & 3 & 4  \\
\hline
1 & 1 & 2 & 5 & 8 & 11 \\
2 & 2 & 4 & 15 & 46 & $\geq 213$ \\
3 & 3 & 9 & $\geq 110$ &  &\\
4 & 4 & 30 & & & \\
5 & 5 & $\geq 122$ &&& \\
\hline
\end{tabular}
\end{center}
\caption{Optimal values of $R(k,t)$}
\label{tab3}
\end{table}

Words achieving the nontrivial bounds in Table~\ref{tab3} are
given below:
\begin{table}[H]
\begin{center}
\begin{tabular}{c|c|l}
$k$ & $t$  & Lexicographically least word achieving $R(k,t)$ \\
\hline
2 & 1 & $0011$ \\
2 & 2 & $010001100111001$ \\
2 & 3 & $0010100110100011111000111010000011101010001100$\\
3 & 1 & $012010210$ \\
4 & 1 & $012031231032021030231321023013$
\end{tabular}
\end{center}
\caption{Lexicographically least word achieving the bounds in Table~\ref{tab2}}
\end{table}

\section{Remarks}

We currently do not know whether the upper bound in Theorem~\ref{thm3} is tight, or asymptotically tight, except when $n \leq 3$.     Improvement of this bound, or construction of examples nearly matching the bound, would be of interest.

It is a challenging computational problem to compute more values
of $C(k,n)$, $S(k,t)$, and $R(k,t)$, which we leave to the reader.

We are grateful to Farbod Yadegarian, who computed lower bounds on
$R(5,1)$, $R(3,2)$, and $R(2,4)$.

\end{document}